\documentclass[conference]{IEEEtran}

\usepackage[letterpaper, left=1.02in, right=1.01in, bottom=1in, top=0.75in]{geometry}
\usepackage{cite}
\usepackage{amsmath,amssymb,amsfonts}
\usepackage{algorithm}
\usepackage{algorithmic}
\usepackage{graphicx}
\usepackage{textcomp}
\usepackage{enumerate}
\newtheorem{lemma}{\rm\textbf{Lemma}}
\newenvironment{proof}{{\noindent\it Proof:}\quad}{\hfill\rule{1.5mm}{3mm} }

\def\BibTeX{{\rm B\kern-.05em{\sc i\kern-.025em b}\kern-.08em
    T\kern-.1667em\lower.7ex\hbox{E}\kern-.125emX}}
\begin{document}

\title{Distributed Layered Grant-Free \\Non-Orthogonal Multiple Access \\for Massive MTC}


\author{\IEEEauthorblockN{Hui Jiang, Qimei Cui, Yu Gu, Xiaoqi Qin, Xuefei Zhang, Xiaofeng Tao}
\IEEEauthorblockA{National Engineering Laboratory for Mobile Network Technologies\\
Beijing University of Posts and Telecommunications, Beijing, 100876, China \\
Email:
{\{jianghui, cuiqimei\}@bupt.edu.cn}
}}
\maketitle
\begin{abstract}
Grant-free transmission is considered as a promising technology to support sporadic data transmission in massive machine-type communications (mMTC). Due to the distributed manner, high collision probability is an inherent drawback of grant-free access techniques. Non-orthogonal multiple access (NOMA) is expected to be used in uplink grant-free transmission, multiplying connection opportunities by exploiting power domain resources. However, it is usually applied for coordinated transmissions where the base station performs coordination with full channel state information, which is not suitable for grant-free techniques. In this paper, we propose a novel distributed layered grant-free NOMA framework. Under this framework, we divide the cell into different layers based on predetermined inter-layer received power difference. A distributed layered grant-free NOMA based hybrid transmission scheme is proposed to reduce collision probability. Moreover, we derive the closed-form expression of connection throughput. A joint access control and NOMA layer selection (JACNLS) algorithm is proposed to solve the connection throughput optimization problem. The numerical and simulation results reveal that, when the system is overloaded, our proposed scheme outperforms the grant-free-only scheme by three orders of magnitude in terms of expected connection throughput and outperforms coordinated OMA transmission schemes by 31.25\% with only 0.0189\% signaling overhead of the latter.

\end{abstract}

\begin{IEEEkeywords}
Grant-free, distributed layered NOMA, hybrid transmission, massive MTC, IoT.
\end{IEEEkeywords}

\IEEEpeerreviewmaketitle
\section{Introduction}

Unlike human-to-human communications, which involve a small number of devices with high-rate and large-sized data, massive MTC (mMTC) is generally characterized by massive MTC devices (MTCDs) with sporadic transmission and low computational capability. Thus channel access is a serious challenge for mMTC. In conventional grant-based communication networks, users access the network via a four-step random access (RA) procedure. In the scenario of mMTC, it is inefficient to establish dedicated bearers for data transmission, since signaling overheads for coordination are proportional to the number of devices. 
Therefore, grant-free transmission is a promising enabler to mMTC. 

Grant-free is gaining a lot of attention recently, as it allows devices to transmit without waiting for the base station (BS) to grant them radio resources \cite{grant-free1}. In \cite{refer1}\cite{refer2}, contention-based transmission technologies are proved to be prevailing. In \cite{b27}, uncoordinated access schemes are shown to be perfect for mMTC due to their low signaling overhead. Conventionally, slotted ALOHA \cite{ALOHA} is used for uplink grant-free, which is based on orthogonal multiple access (OMA). However, it seriously suffers from the nuisance of collision, which is due to contention based access by multiple devices. Fortunately, by exploiting power domain, non-orthogonal multiple access (NOMA) enables multiple users to share one time-frequency resource. Therefore, NOMA based grant-free can support a significantly increased connections \cite{b32}\cite{b26}. However, most of existing studies on NOMA focus on coordination with known channel state information (CSI) at both transmitter and receiver sides, to optimize subchannel and power allocation \cite{b29} \cite{b30}, which is not suitable for grant-free transmissions; there are also several works regarding grant-free with code domain NOMA \cite{b34}\cite{b41}, which employ various compressive sensing (CS) techniques for multi-user detection (MUD). The main limitation of these works is that prior information about user activity is required, resulting in high computational complexity on the receiver side. Moreover, the proposed approaches are only suitable for the cases where user activity is time-related and sparse.

To address these challenges, we adopt a distributed NOMA, power division multiple access\cite{b16}, and propose a low-complexity distributed layered grant-free NOMA framework and a hybrid transmission scheme, accordingly. Under this framework, the inherent drawback of grant-free random access, high collision probability due to its distributed manner, can be greatly alleviated. The key of the proposed framework is dividing the cell into different layers based on predetermined inter-layer received power difference, thus, power domain NOMA can be used to drastically reduce the number of MTCDs that compete for grant-free transmission in each region. Instead of transmitting on the allocated subchannels with allocated transmission power,  MTCDs in each region can decide their own transmission power and subchannels for direct data transmission without the BS assistance. With predetermined received powers, the computation complexity of MUD on the receiver side is decreased. To further guarantee the connection throughput no matter what the system load is, we apply Enhanced Access Barring (EAB) mechanism \cite{b21} for adaptive congestion control of the framework. Based on the framework, a hybrid transmission scheme is also to improve the transmission success probability and reduce overhead, significantly.

In this paper, we propose a novel distributed layered grant-free NOMA framework and a hybrid transmission scheme. To efficiently characterize the system performance, we derive a closed-form analytic expression for connection throughput, based on the number of NOMA power levels, available subchannels, and the number of contenders. The accuracy of the expression is validated through Monte Carlo simulation. A joint access control and NOMA layer selection (JACNLS) algorithm is proposed to solve the connection throughput optimization problem. Numerical analysis and simulation results prove that our proposed scheme outperforms the grant-free-only scheme by three orders of magnitude in terms of expected connection throughput, and outperforms coordinated transmission schemes by 31.25\% with only 0.0189\% signaling overhead.

The main contributions of this paper are as follows.
\begin{itemize}
\item We propose a novel distributed layered grant-free NOMA framework. The cell is divided into layers based on predetermined inter-layer received power difference. MTCDs transmit distributedly according to the proposed hybrid transmission scheme. Note that the parameter for dividing cell is obtained by solving a throughput optimization problem.

\item We propose a distributed layered grant-free NOMA based hybrid transmission scheme. Moreover we derive a closed-form analytic expression for connection throughput, which can effectively characterize the system performance. 
\item We formulate a connection throughput optimization problem, and propose a JACNLS algorithm to find the optimal parameters. 
\end{itemize}

The rest of the paper is organized as follow. In Section II, we introduce the system model. In Section III, we propose a novel distributed layered grant-free NOMA framework and analyze the connection throughput. Section IV we solve a connection throughput optimization problem. In Section V, we present numerical and simulation results and the paper is concluded in Section VI.
\section{System Model}
Consider a cellular network as shown in Fig. \ref{fig1}, with a single BS located in the origin serving $Q$ MTCDs in the area. We assume that the MTCDs are uniformly distributed in a circle of radius $D$.  All MTCDs share a bandwidth of $B_{T}$ for uplink data transmissions. The available system bandwidth is divided into frequency resource blocks, each of bandwidth $B$. Thus, the total number of frequency resource blocks is given as $M=B_{T}/B$. We evaluate the performance of the system within the time slot period $T_{\rm{P}}$. 

In this paper, we use \emph{Connection Opportunity} (CO) to represent a connection resource. The number of COs in a time slot is determined only by available subchannels in OMA systems. Due to the limited frequency spectrum, the number of COs is inadequate for massive grant-free access. As seen in Fig. \ref{fig1}(a), if two users in a cell access the BS with the same subchannel in a time slot, collision happens. 

\begin{figure}[htbp]
\centerline{
\includegraphics[width=0.48\textwidth]{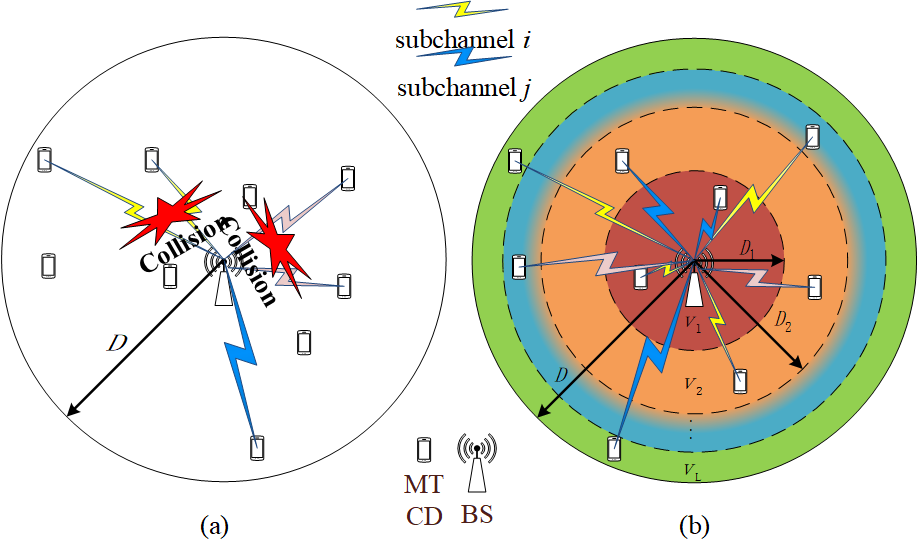}}
\caption{(a) Grant-free transmission; (b) principle diagram of a distributed layered NOMA, where different colored circle and rings indicate the layers of different aiming received powers, and darker color represents bigger received power.}

\label{fig1}
\end{figure}
A distributed NOMA concept is applied. Suppose that there are \emph{L} predetermined aiming received power levels that are denoted as $v_{1}>v_{2}>...>v_{L}>0$, where ${v_l} = \Gamma {(\Gamma  + 1)^{L - l}}$($l=1,2,...,L$), and $\Gamma$ is the target signal to interference-plus-noise ratio (SINR). The single-BS cell is divided into $L$ concentric layers, MTCDs in different layers have different aiming received power, as shown in fig .\ref{fig1}(b), the insider layer denotes the smaller received power. MTCDs decide their transmission power according to locations and CSI. For example, for an MTCD $k$, $d_{k}$ is the distance to the BS, if it belongs to set ${{\rm K}_l} = \left\{ {k|{D_{l - 1}} < {d_k} \le {D_l}} \right\}$, where ${D_0} = 0$, ${D_l} = D\sqrt {\frac{l}{L}}$, then its aiming received power level is $v_{l}$. Based on its knowledge of channel gain of different sunchannel $i$, refered as ${g_{i,k}}$ ($i=1,...,M$), its transmission power is decided as ${P_k} = \frac{{{v_l}}}{{\mathop {\max }\limits_i {g_{i,k}}}}$.

The average transmission power for an MTCD using the distributed NOMA when $M \ge 2$ is upper-bounded as\cite{b16}
\begin{equation}
\begin{array}{l}
{P_{{\rm{ave}}}}  \le \frac{{\min \left\{ {2\ln 2,\frac{M}{{M - 1}}} \right\}}}{L}\sum\limits_{l = 1}^L {\frac{{\Gamma {{\left( {\Gamma  + 1} \right)}^{L - l}}}}{{{A_0}{{\left( {D\sqrt {\frac{l}{L}} } \right)}^{ - \beta }}}}} 
\end{array},
\label{eq1}
\end{equation}
where $\beta$ is the path loss exponent, and $A_{0}$ is a constant related to antenna gains.

\section{Design and Analysis of the Distributed Layered Grant-Free NOMA Framework}
\subsection{Framework Design}
Considering the mass number and sporadic transmission characteristics of mMTC, we propose a distributed layered grant-free NOMA framework and a hybrid transmission scheme for uplink mMTC, which is shown in Fig. \ref{fig1}(b). The cell is divided into different layers based on predetermined inter-layer received power difference. The BS does not perform any power and subchannel allocation coordination. An EAB mechanism can effectively control the number of contenders and distribute the traffic over time, the principle is to let the BS broadcast a parameter to all MTCDs. When an MTCD tries to transmit, it generates a random number between 0 and 1, and compares the generated number with the EAB access control parameter ${p_{\rm{E}}}$ (${{\rm{0 < }}{p_{{\rm{E}}}} \le {\rm{1}}}$) broadcast by BS. If the number is less than ${p_{\rm{E}}}$, the MTCD proceeds to transmit. Otherwise, it needs to backoff temporarily. For the BS, to broadcast the signaling containing: \textcircled{1} EAB access control parameter ${p_{\rm{E}}}$; \textcircled{2} the number of NOMA layer parameter $L$, which is used to divide the cell, an optimization operation is performed with the information about current number of MTCDs competing for grant-free transmission and spectrum resources; For MTCDs, after with fast retrial\cite{b19} and infinite retransmissions assumed, a hybrid transmission scheme is proposed. Instead of transmitting on the allocated subchannels with allocated transmission power, MTCDs in each NOMA layer can distributedly decide their own transmission power and subchannels for direct data transmission based on their location and CSI. The detailed steps are as follows.

\begin{enumerate}[step 1]
\item Subsequent to the system initialization, the BS performs an optimization operation to get EAB access control parameter ${p_{\rm{E}}}$ and the number of NOMA layers $L$ in each time slot, and broadcasts them.
\item Once there is a packet, the MTCD generates a random number $p$, if $p$ is smaller than $p_{{\rm{E}}}$, it calculates its transmission power independently based on its location, CSI (mainly channel gain) and $L$, then go to step 3; otherwise, it waits for next available time slot;
\item Each MTCD transmits with the calculated transmission power for direct data transmission;
\item If \emph{DL ACK} signal, which denotes the data is decoded by the receiver successfully, is received by the MTCD within the prescribed period, the transmission is successful; otherwise, it tries in next available time slot. 
\end{enumerate}

This framework significantly reduces computational complexity at the BS, as channel estimations of massive MTCDs is unnecessary. Furthermore, with the predetermined received powers, the computation complexity of MUD on the receiver side is greatly decreased.

The signaling procedure of connection oriented and the proposed hybrid transmission schemes is shown in fig. 2, when an MTCD has bursty small data. The former causes heavy signaling overhead up to 220 bytes \cite{b31} even if for a tens-of-bytes packet, however, the latter can reduce the signaling overhead to  only a few bytes for broadcasting signal.

\begin{figure}[htbp]
\centerline{
\includegraphics[width=0.45\textwidth]{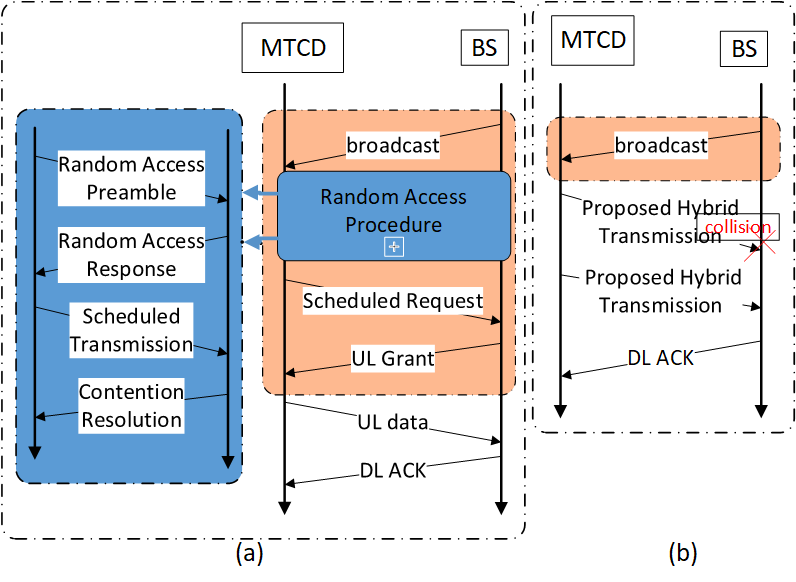}}
\caption{Comparison between (a) connection oriented communication and (b) proposed hybrid access and data transmission.}
\label{fig2}
\end{figure}

\subsection{Analysis of Connection Throughput}
We define a concept of \emph{connection throughput} to represent the theoretical average number of MTCDs which can successfully transmit without collision, given the resources provided in a time slot. 
Suppose there are $M$ subchannels and $C_{T}$ active MTCDs, under the distributed layered grant-free NOMA framework, it is equal to the situation that there are $L$ layers, each layer has $M$ COs. To simplify mathematics, we assume that the number of active MTCDs in each layer is the same (i.e. $C = \frac{{{C_T}}}{L}$), since the MTCDs are uniformly distributed. The assumption can be easily extended to the general scenario with different number of MTCDs for each layer. For a subchannel, if there are multiple MTCDs who choose the same power level, the signals cannot be decoded, which is called power collision. The power collision at each power level in NOMA is not an independent event, that is, if it happens at level $l$, it does not affect the signals at levels $1,...,l-1$, but the signals at levels $l + 1,...,L$ cannot be decoded, even though these higher level signals have no power collision themselves. Considering receiver decoding sequence, that is, the bigger received power, the earlier decoding, and error propagation, the connection probability of $l$-th layer is determined as follows.
\begin{lemma}
\emph{The connection probability of $l$-th layer for given $C$ contenders and $M$ subchannels under the proposed distributed layered grant-free NOMA framework, denoted by $P_l^{{\rm{Con}}}$, is determined as}
\begin{equation}
P_l^{{\rm{Con}}} = {\left(1 - \frac{1}{M}\right)^{Cl - 1}}{\left(1+\frac{C}{{M - 1}}\right)^{l - 1}}.
\label{eq2}
\end{equation}
\end{lemma}
\begin{proof}
Let $p_{c,m,l}^{{\rm{sel}}}$ denotes the probability that a random MTCD $c$ in the $l$-th layer uniquely select subchannel $m$, it is calculated as $p_{c,m,l}^{{\rm{sel}}} = \frac{1}{M}{(1 - \frac{1}{M})^{C - 1}}$; $p_{m,l}^{{\rm{nosel}}}$ denotes the probability that no MTCD in the $l$-th layer selects subchannel $m$, calculated as $p_{m,l}^{\rm{nosel}} = {(1 - \frac{1}{M})^C}$; $p_{c,m,l}^{\rm{succ}}$ denotes the probability that a random MTCD $c$ in the $l$-th layer uniquely selects subchannel $m$ and decode successfully at the receiver side. 

For a random MTCD belongs to set $\rm{K}_{1}$, the event that it uniquely selects subchannel $m$ and decode successfully at the receiver side is independent of other MTCDs with signals at levels $2,3,..,L$, since it is decoded first at the receiver side. So
\begin{equation}
p_{c,m,1}^{\rm{succ}} = p_{c,m,1}^{\rm{sel}}.
\end{equation}

For an MTCD belongs to set $\rm{K}_{2}$, the event that it is decoded successfully at the receiver side happens only when it is the only one who selects subchannel $m$ in set $\rm{K}_{2}$, and the signal over subchannel $m$ in set $\rm{K}_{1}$ can be decoded successfully or no MTCD in in set $\rm{K}_{1}$ selects subchannel $m$. So 
\begin{equation}
p_{c,m,2}^{\rm{succ}} = p_{c,m,2}^{\rm{sel}}\left( {\sum\limits_{c = 1}^C {p_{c,m,1}^{\rm{succ}}} {\rm{ + }}p_{m,1}^{\rm{nosel}}} \right),
\end{equation}
then 
\begin{equation}
\sum\limits_{c = 1}^C {p_{c,m,1}^{\rm{succ}}} {\rm{ + }}p_{m,1}^{\rm{nosel}} = \frac{{p_{c,m,2}^{\rm{succ}}}}{{p_{c,m,2}^{\rm{sel}}}}.
\end{equation}
And for an MTCD belongs to set $\rm{K}_{3}$, 
\begin{equation}
\begin{array}{l}
p_{c,m,3}^{\rm{succ}} = p_{c,m,3}^{\rm{sel}}\left(\sum\limits_{c = 1}^C {p_{c,m,2}^{\rm{succ}}} {\rm{ + }}p_{m,2}^{\rm{nosel}}\left(\sum\limits_{c = 1}^C {p_{c,m,1}^{\rm{succ}}} {\rm{ + }}p_{m,1}^{\rm{nosel}}\right)\right)\\
\;\;\;\;\;\;\;\; = p_{c,m,3}^{{\rm{sel}}}\left(\sum\limits_{c = 1}^C {p_{c,m,2}^{\rm{succ}}} {\rm{ + }}p_{m,2}^{\rm{nosel}}\cdot\frac{{p_{c,m,2}^{\rm{succ}}}}{{p_{c,m,2}^{\rm{sel}}}}\right)\\
\;\;\;\;\;\;\;\; = p_{c,m,3}^{\rm{sel}}\cdot p_{c,m,2}^{\rm{succ}}\left(C{\rm{ + }}\frac{p_{m,2}^{\rm{nosel}}}{{p_{c,m,2}^{\rm{sel}}}}\right).
\end{array}
\end{equation}
Then we can find that
\begin{equation}
p_{c,m,l}^{\rm{succ}} = p_{c,m,l}^{\rm{sel}}\cdot p_{c,m,l - 1}^{\rm{succ}}(C{\rm{ + }}\frac{p_{m,l - 1}^{\rm{nosel}}}{p_{c,m,l - 1}^{\rm{sel}}}),\;l = 2,3,...,L,
\end{equation}
which is a recursive expression of $p_{c,m,l}^{\rm{succ}}$, that is 
\begin{equation}
\begin{array}{l}
p_{c,m,l}^{\rm{succ}} = \frac{1}{M}{\left(\frac{M-1}{M}\right)^{C - 1}}\left(C + \frac{{{{\left(1 - \frac{1}{M}\right)}^C}}}{{\frac{1}{M}{{\left(1 - \frac{1}{M}\right)}^{C - 1}}}}\right)p_{c,m,l-1}^{\rm{succ}}
\\
\;\;\;\;\;\;\;\;\; = A \cdot p_{c,m,l - 1}^{\rm{succ}}, 
\end{array}
\end{equation}
It can be written as $Q(l) = A \cdot Q(l - 1)$, where $A = \frac{{C + M - 1}}{M}{\left(1 - \frac{1}{M}\right)^{C - 1}}$. This is the form of isometric series, the formula of general term is $Q(l) = Q(1){A^{l - 1}}$, where $Q(1) = p_{c,m,1}^{\rm{succ}} = \frac{1}{M}{\left(1 - \frac{1}{M}\right)^{C - 1}}$. Then 
\begin{equation}
\begin{array}{l}
p_{c,m,l}^{\rm{succ}} = \frac{1}{M}{\left(1 - \frac{1}{M}\right)^{C - 1}}{\left(\frac{{C + M - 1}}{M}{\left(1 - \frac{1}{M}\right)^{C - 1}}\right)^{l - 1}}\\
\;\;\;\;\;\;\;\;\; = \frac{{{{\left(M - 1\right)}^{\left(C - 1\right)l}}{{\left(C + M - 1\right)}^{l - 1}}}}{{{M^{Cl}}}}.
\end{array}
\end{equation}
So the probability that a random MTCD $c$ in the $l$-th layer can be decoded successfully at the receiver side, that is, the connection probability of $l$-th layer is determined as
\begin{equation}
\begin{array}{l}
P_l^{{\rm{Con}}} = \sum\limits_{m = 1}^M {p_{c,m,l}^{\rm{succ}}}  = \frac{{{{\left(M - 1\right)}^{\left(C - 1\right)l}}{{\left(C + M - 1\right)}^{l - 1}}}}{{{M^{Cl - 1}}}}\\
\;\;\;\;\;\;\;\;\;\;\;\; = {\left(1 - \frac{1}{M}\right)^{Cl - 1}}{\left(1 + \frac{C}{{M - 1}}\right)^{l - 1}}.
\end{array}
\end{equation}
\end{proof}

So the connection throughput of $l$-th layer is determined as $T_l^{{\rm{Con}}} = C\cdot{P_l^{{\rm{Con}}}}$.
\subsection{Validation of Connection Throughput}
In this section, we will present that the proposed concept of connection throughput can be used as a performance metric to evaluate the system performance under different amount of MTCDs. Fig. \ref{fig4} presents the numerical and simulation results of connection throughput of each layer MTCDs in Eq. \eqref{eq2}. The upper results are connection throughput when the system is not overloaded ($C_{T}=200$), while the lower results present the low efficiency when the system is overloaded ($C_{T}=500$). We can find that, the bigger the number of MTCDs is, the lower the connection probability is, which is in line with the actual situation. Therefore, our defined connection throughput can reflect the performance of the system accurately, so can be used as a metric of system performance. We will use numerical results to evaluate the connection throughput performance of the system in Section V.

\begin{figure}[htbp]
\centerline{
\includegraphics[width=0.4\textwidth]{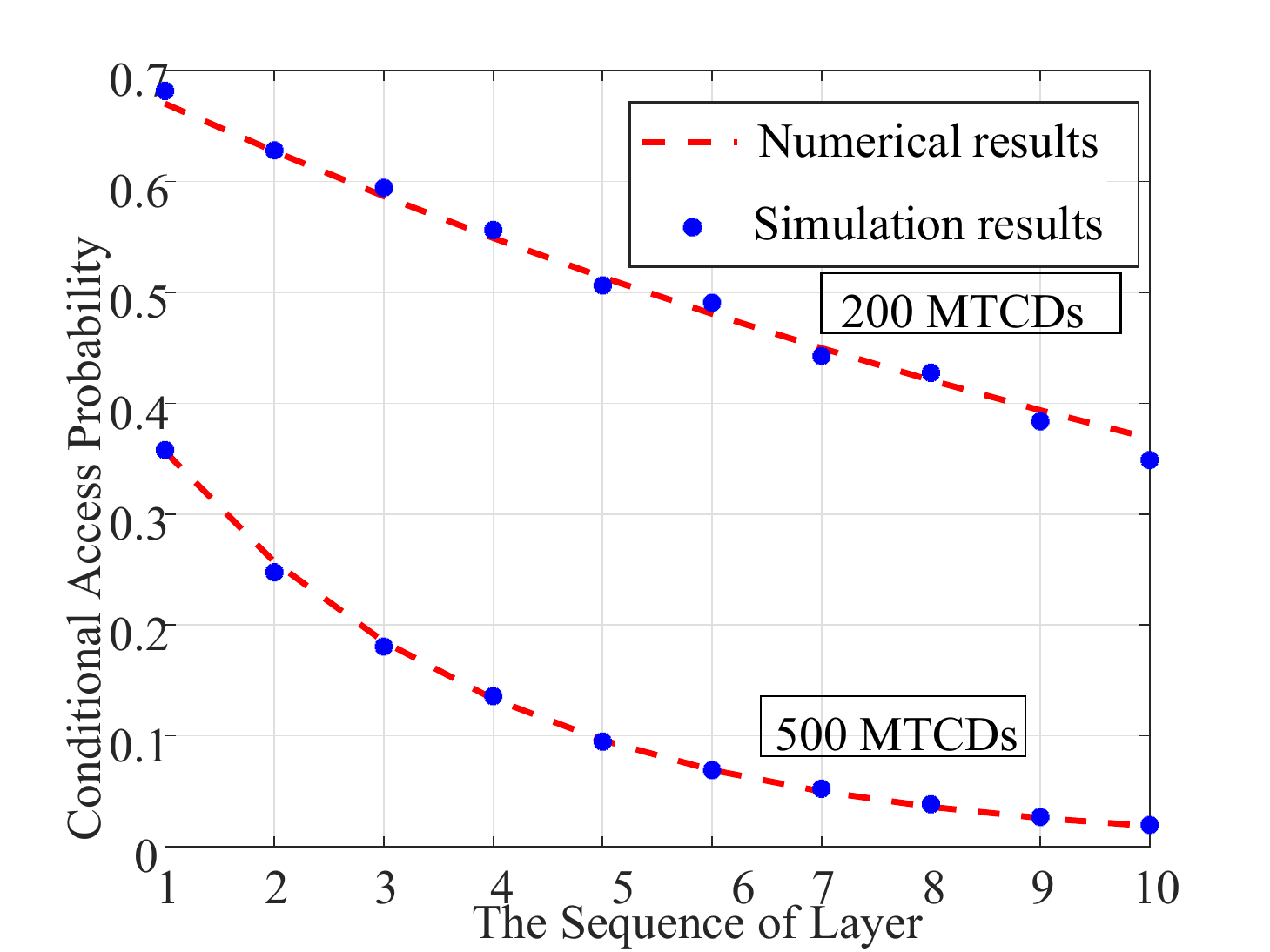}}
\caption{Connection probability simulation and numerical results.}
\label{fig4}
\end{figure}

\section{Connection Throughput Optimization}
\subsection{Problem Formulation}
In this section, we study a connection throughput maximization problem for the distributed layered grant-free NOMA framework, with average access delay constraint and maximum allowed average transmission power constraint. Assuming $Q$ is the total number of MTCDs before EAB access control, and access control parameter is ${p_{\rm{E}}}$, then $C_T = {p_{\rm{E}}} \cdot Q$ MTCDs are allowed to attempt access. With the proposed scheme, the number of contenders in each layer is $C = \frac{{{p_{\rm{E}}} \cdot Q}}{L}$. By applying fast retrial and infinite retransmission mechanism, the average required transmission time follows Bernoulli trial, as given by $T_{\rm{ave}} = \sum\limits_{r = 0}^\infty  {\left( {r + 1} \right){p_{a}}{{\left( {1 - {p_{a}}} \right)}^r}}  = \frac{1}{{{p_{a}}}}$, where $r$ denotes the retransmission time, ${p_{a}}$ denotes the probability that an MTCD joins in the competition in a time slot when the EAB mechanism is adopted. More specifically, ${p_{a}} = {p_{\rm{E}}} \cdot {p_{\rm{succ}}}$, 
where ${p_{\rm{succ}}} = \frac{1}{L}\sum\limits_{l = 1}^L {p_{l}^{\rm{succ}}} $, and $p_{l}^{\rm{succ}}$ is the probability that an MTCD of $l$-th layer successfully gains access when it joins in the competition, which is given as Eq. \eqref{eq2}. Then the average access delay of MTCDs can be determined as
\begin{equation}
{D_{{\rm{ave}}}} = {T_{{\rm{P}}}} \cdot T_{\rm{ave}} = \frac{{{T_{{\rm{P}}}}}}{{{p_a}}},
\end{equation}



To avoid infinite retransmissions and over-limit transmission power, average access delay constraint and maximum allowed average transmission power constraint should be imposed while maximizing connection throughput. 
Then the problem of maximizing the connection throughput can be formulated as
\begin{equation}
\begin{array}{l}
\mathop {\max }\limits_{L,{p_{\rm{E}}}}  \;T^{{\rm{Con}}} = \sum\limits_{l = 1}^L T_l^{{\rm{Con}}} \\
\rm{s.t.}\;\;\;C1:L \in {{\rm{L}}_{\rm{T}}} = \left\{ {{\rm{1,2,}}...{\rm{,}}{L_{\max }}} \right\},\\
\;\;\;\;\;\;\;\;C2:{D_{{\rm{ave}}}} \le {D_{{\rm{req}}}},
\end{array}
\end{equation}
where ${\rm{L_T}}$ is the set of acceptable number of NOMA power level. The maximum acceptable number of power levels ${L_{\max }}$ is determined by ${L_{\max }} = \max \{ {L_u},\max \{ L|{P_{{\rm{ave}}}} < {P_{\max }},L \in Z\} \} $, where ${P_{{\rm{ave}}}}$ denotes the average transmission power for a random MTCD using distributed NOMA scheme, which is calculated in Eq. \eqref{eq1}, and ${P_{\max }}$ denotes the maximum allowed average transmission power, $L_{u}$ denotes the receiver complexity allowed maximum power level, and $D_{\rm{req}}$ is the minimum delay requirement. $C1$ and $C2$ denote the constraints of maximum allowed average transmission power and average access delay, respectively. 

\subsection{Joint Access Control and NOMA Power Level Selection Algorithm}

Notice that connection throughput $T^{{\rm{Con}}}$ is an increasing function of number of NOMA power level $L$, then $L = {L_{\max }}$ is always chosen. The feasible set $\rm P$ of ${p_{\rm{E}}}$ to meet the average delay requirement can be derived as
\begin{equation}
{\rm P} = \left\{ {{p_{{\rm{E}}}}\left| {\begin{array}{*{20}{c}}
{\frac{{{T_{{\rm{P}}}}}}{p_a} \le {D_{{\rm{req}}}},}\\
{0 < {p_{{\rm{E}}}} \le 1}
\end{array}} \right.} \right\},
\end{equation}
where $p_a={\frac{{{p_{{\rm{E}}}}}}{{{L_{\max }}}}\sum\limits_{l = 1}^L {{{(1 - \frac{1}{M})}^{\frac{{Q \cdot {p_{{\rm{E}}}}}}{{{L_{\max }}}}l - 1}}{{(1 + \frac{{\frac{{Q \cdot {p_{{\rm{E}}}}}}{{{L_{\max }}}}}}{{M - 1}})}^{l - 1}}} }$. Finally, the optimal ${p_{\rm{E}}}$ is given by
\begin{equation}
{p_{\rm{E}}}^* = \arg \;\mathop {\max }\limits_{{p_{\rm{E}}} \in {\rm{P}}} {T^{{\rm{Con}}}}.
\end{equation}
If $\rm P$ is an empty set, there is no optimal solution, which denotes that, even the maximum acceptable number of power levels and optimal access control parameter are adopted, the maximized the connection throughput is still inadequate for such a big number of contenders. In this case, we can determine ${p_{\rm{E}}}^*$ as 
\begin{equation}
{p_{{\rm{E}}}}^* = \arg \;\mathop {\max }\limits_{{\rm{0 < }}{p_{{\rm{E}}}} \le {\rm{1}}} {T^{{\rm{Con}}}},
\end{equation}

which targets the maximum connection throughput but exceeds the delay requirement. However, our JACNLS algorithm gives the best-effort scheme for the given contender number, which can shorten the average delay to the maximum degree. The details are shown in Algorithm \ref{alg:A}, which find the optimal parameters to obtain the optimal connection throughput.

\begin{algorithm}[htb]
\caption{JACNLS Algorithm}
\label{alg:A}
\begin{algorithmic}[1] 
\REQUIRE ~~\\ 
The number of contenders, $Q$;\\
The number of subchannels available, $M$;
\ENSURE ~~\\ 
Access control parameter, $p_{E}$;\\
NOMA power level, $L$;

\STATE {calculate \\${L_{\max }} = \max \{ {L_u},\max \{ L|{P_{{\rm{ave}}}} < {P_{\max }},L \in Z\} \} $}
\IF{${L_{\max }} \leq L_{u} $}   
\STATE $L=L_{\max}$   
\ELSE   
\STATE $L=L_{u}$   
\ENDIF 
\FOR{each $p_{{\rm{E}},i} \in (0,1]$}  
\STATE calculate $p_{i}^{\rm{succ}}$ and $T_{i}^{{\rm{Con}}}$;
\ENDFOR 
\STATE {set ${p_{\rm{E}}} = \arg \;\mathop {\max }\limits_{{p_{{\rm{E}},i}} \in (0,1]} T_{i}^{{\rm{Con}}}$}
\IF{the selected ${{{p_{{\rm{E}}}} \cdot {p^{{\rm{succ}}}}}} \geq {T_{{\rm{P}}}}/{D_{{\rm{req}}}}$}   
\STATE blue ${T^{{\rm{Con}}}} = \max \{ {T_{i}^{{\rm{Con}}}}\} $   
\ELSE   
\STATE red ${T^{{\rm{Con}}}} = \max \{ {T_{i}^{{\rm{Con}}}}\} $   
\ENDIF 
\RETURN $p_{\rm{E}}$ and $L$; 
\end{algorithmic}
\end{algorithm}

\section{Numerical Analysis and Performance Evaluation}
\begin{table}[htbp]
\caption{Notation Summary}
\begin{center}
\begin{tabular}{|c|c|c|}
\hline
Notation & Description & Value\\ \hline
$B_{T}$ & Total bandwidth of the system in Hz & 180\\ \hline
$B$ & The bandwidth of subcarrier in Hz & 3.75 \\ \hline
$M$ & The number of subchannels & 48 \\ \hline
$\Gamma$ &  Target SINR in dB & $6$\\ \hline
${D_{\rm{req}}}$ &  Minimum delay requirement in ms & $1$ \\ \hline
${P_{\max }}$ &  Maximum transmission power in dBm & $18$\\ \hline
$T_{\rm{P}}$ & time slot period in ms & $0.2$\\ \hline
$L_{u}$ & Allowed maximum power level & $5$ \\ \hline
\end{tabular}
\end{center}
\label{tab1}
\end{table}
%
In this section, we present numerical 
results to evaluate the performance of proposed hybrid transmission scheme, with JACNLS algorithm to maximize the connection throughput. 
The list of key mathematical symbols used in this paper are summarized in Table I. For the path loss exponent $\beta$, we assume that $\beta$ = 3.8. In addition, $D=1$ and $A_0 = 1$ used in Eq. \eqref{eq1} are assumed for normalization purpose. The maximum allowed average transmission power ${P_{\max }}$ is set to 18dBm since the maximum transmission power is set to 23dBm in general\cite{b33}.

When the number of subchannels $M \ge 2$, the average transmission power for a random MTCD using the proposed hybrid transmission scheme, a distributed NOMA scheme, is upper-bounded as Eq. \eqref{eq1}. Other schemes such as NOMA with random subchannel and power level selection (random NOMA), coordinated OMA and so on, are presented for comparison purposes. For random NOMA scheme, if the subchannel and power level are randomly selected, the average transmission power would be ${P_{\rm{ave\_2}}}  = \frac{1}{L}\sum\limits_{l = 1}^L {{v_l}{\rm E}\left[ {\frac{1}{{{g_{i,k}}}}} \right]}$, 
where ${g_{i,k}}$ is the channel power gain from MTCD \emph{k} to the BS over subchannel $i$. $L_{\max\_2}$ is determined by ${L_{\max\_2}}=\max\{{L_u},\max\{L|{P_{{\rm{ave{\_2}}}}}<{P_{\max }},L\in Z\}\}$. 

From Fig. 4(a), the average transmission power of the proposed hybrid transmission scheme grows slower than that of random NOMA as $L$ increases, denoting the maximum acceptable number of NOMA layers of the proposed scheme is much bigger than that of random NOMA. Since the maximum allowed average transmission power ${P_{\max }}$ is 18dBm, then the maximum acceptable number of power levels ${L_{\max }}$ in hybrid transmission and random NOMA are 5 and 3, respectively; Fig. 4(b) shows the average transmission power for different numbers of subchannels. As expected, the average transmission upper-bound power of our scheme decreases with $M$ increasing. However, the average transmission power of random NOMA does not depend on $M$. That is, a large $M$ can help improve energy efficiency of the proposed scheme. This encourages the use of narrowband transmission if limited spectrum resources are given. So, for mMTC, we set $M=48$; When hybrid transmission scheme is applied, the required signaling overhead is from broadcasting signals containing NOMA power level $L$ and optimal access control parameter $p_{\rm{E}}$, about 2 bytes, while with conventional connection setup procedures, transmitting a small sized data packet needs 220 bytes. Fig. 4(c) shows that our scheme will drastically reduce the signaling overhead to 0.0189\% comparing to coordinated schemes when 48 MTCDs access successfully.
\begin{figure}[htbp]
\centerline{
\includegraphics[width=0.5\textwidth]{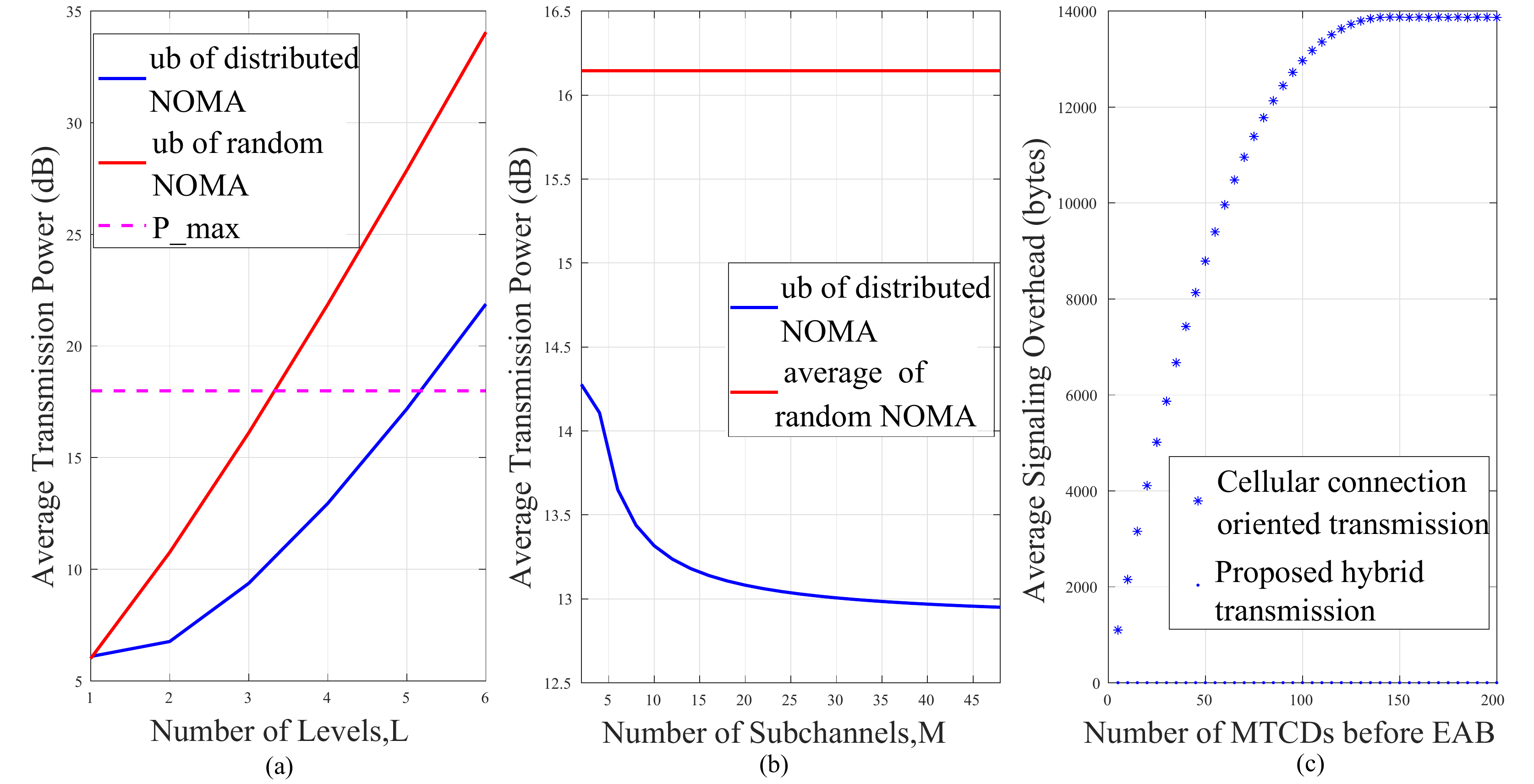}}
\caption{(a) Average transmission power upper-bound (ub) for different values of $L$, $P_{\rm{max}}$ represents the maximum average transmission power; (b) average transmission power of random NOMA and ub of distributed NOMA for different values of $M$; (c) signaling overhead comparison between cellular connection oriented transmission and proposed hybrid transmission for different values of $Q$.}
\label{fig3}
\end{figure}

\begin{figure}[htbp]
\centerline{
\includegraphics[width=0.4\textwidth]{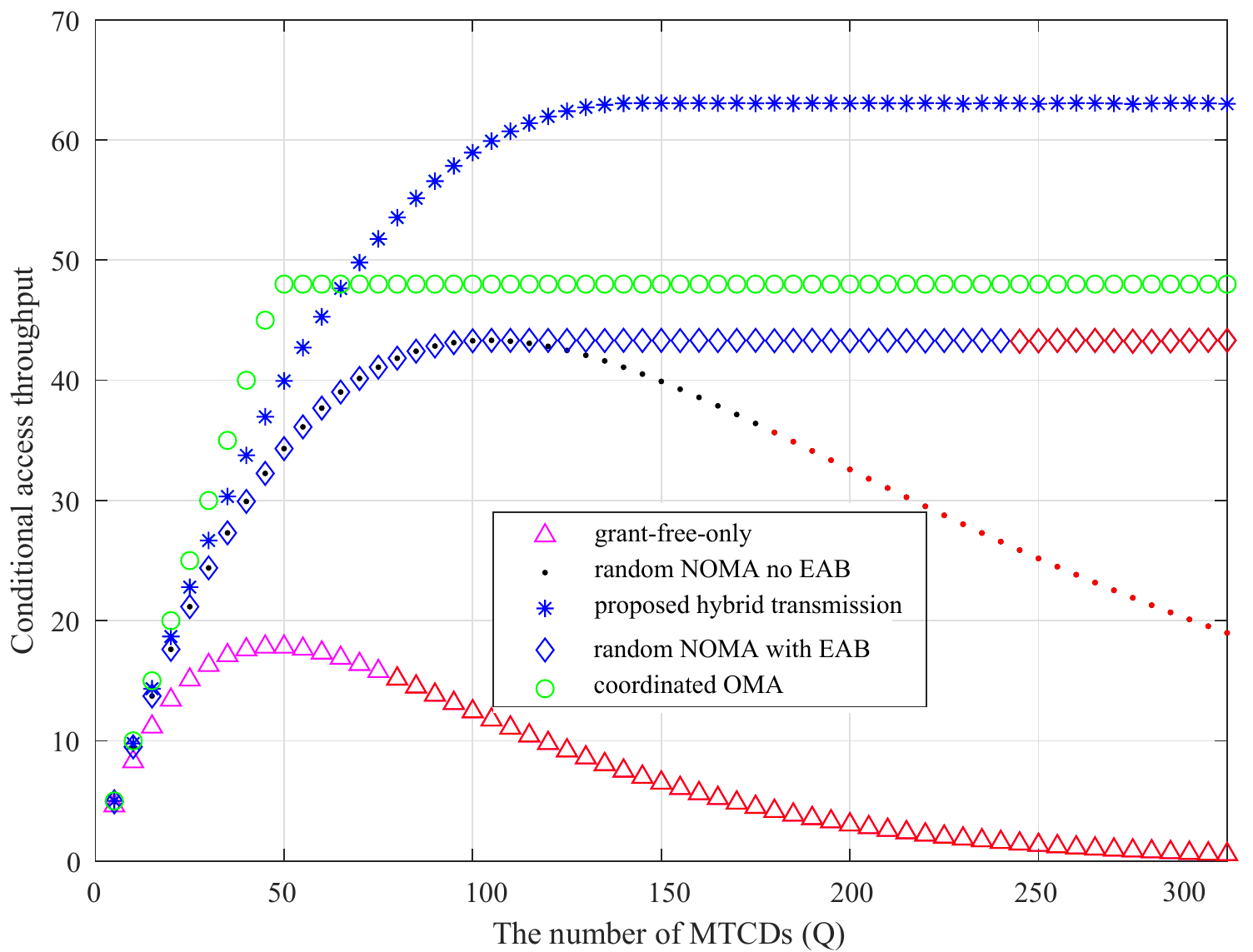}}
\caption{Connection throughput comparison between bybrid transmission, random NOMA with/without EAB, grant-free-only scheme, for different values of $Q$.}
\label{fig5}
\end{figure}

As shown in Fig. 5, color non-red signs represent the maximum connection throughput is obtained under the constraint of required average access delay, while the color red signs don't, but still obtain its maximum. It is obvious that the performance of connection throughput based on our proposed scheme is better than that any other schemes listed, especially when $Q$ is big. 

In Fig. 5, when the number of contenders is 300, our proposed hybrid transmission scheme outperforms high-complexity\&overhead coordinated OMA, random NOMA with EAB, random NOMA without EAB and grant-free-only schemes by about 31.25\%, 46.6\%, 231.6\% and 6,200\% in terms of connection throughput, while satisfying average access delay requirement. 

\section{Conclusion}
In this work, to support more connectivity in uplink grant-free mMTC, we proposed a novel distributed layered grant-free NOMA framework based on distributed NOMA. The proposed hybrid transmission scheme can significantly reduce signaling overhead to 0.0189\% comparing to coordinated schemes. 
In addition, we have derived a closed-form analytic expression for the expected connection throughput of the transmission scheme. A JACNLS algorithm has been proposed to maximize the connection throughput and to resolve the collision problem. 
The numerical analysis and simulation results reveal that when the system is overloaded, our proposed scheme outperforms the grant-free-only scheme by three orders of magnitude, and outperforms high-complexity and high-overhead coordinated OMA transmission schemes by 31.25\%, in terms of expected connection throughput. Therefore, the proposed distributed layered grant-free NOMA framework and according scheme are suitable for grant-free mMTC when the spectrum resources are limited.

\section*{Acknowledgment}
The work was supported in part by the National Nature Science Foundation of China Project under Grant 61471058, in part by the Hong Kong, Macao and Taiwan Science and Technology Cooperation Projects under Grant 2016YFE0122900, in part by the Beijing Science and Technology Commission Foundation under Grant 201702005, in part by the National Science and Technology Major Project of China under Grant 2017ZX03001004, in part by the Key National Science Foundation of China under Grant 61461136002, and in part by the 111 Project of China under Grant B16006.

\end{document}